\documentclass[11pt]{article}

\usepackage{amsfonts}
\usepackage{amsmath}
\usepackage{graphicx}
\usepackage{hyperref}
\usepackage{amssymb}
\usepackage{amsopn,ntheorem}
\usepackage{comment}
\usepackage[usenames,dvipsnames]{color}

\usepackage[a4paper]{geometry}
\newtheorem{theorem}{Theorem}
\newtheorem{lemma}[theorem]{Lemma}

\newtheorem{proposition}[theorem]{Proposition}

\newtheorem{corollary}[theorem]{Corollary}
\newtheorem{definition}[theorem]{Definition}
\newtheorem{example}[theorem]{Example}

\newenvironment{proof}{{\noindent\emph{Proof:}}}{$\hfill\Box$\vspace{.1in}}

\def\P{\mbox{\rm{P}}}
\def\poly{\mbox{\rm{poly}}}

\def\NP{\mbox{\rm{NP}}}
\def\coNP{\mbox{\rm{coNP}}}

\def\Pr{\mbox{\rm{Pr}}}

\newcommand{\B}{\{0,1\}}

\begin{document}
\title{The Value of Help Bits in Randomized and Average-Case Complexity}

\author{Salman Beigi$^1$, Omid Etesami$^1$, and Amin Gohari$^{1,2}$
\\ \emph{\small $^1$School of Mathematics, Institute for Research in Fundamental Sciences (IPM), Tehran, Iran}\\
\emph{\small $^2$Department of Electrical Engineering, Sharif University of Technology, Tehran, Iran}}

\date{July 6, 2014}

\maketitle

\begin{abstract}
``Help bits" are some limited trusted information about an instance or instances of a computational problem
that may reduce the computational complexity of solving that instance or instances.
In this paper, we study the value of help bits in the settings of randomized and average-case complexity.

Amir, Beigel, and Gasarch (1990) show that for constant $k$, if $k$ instances of a decision problem can be efficiently solved using less than $k$ bits of help, then the problem is in $\P/\poly$.
We extend this result to the setting of randomized computation: We show that the decision problem is in $\P/\poly$ if
using $\ell$ help bits, $k$ instances of the problem can be efficiently solved with probability greater than $2^{\ell-k}$.
The same result holds if using less than $k(1 - h(\alpha))$ help bits (where $h(\cdot)$ is the binary entropy function), we can efficiently solve $(1-\alpha)$ fraction of the instances correctly with non-vanishing probability.
We also extend these two results to non-constant but logarithmic $k$. In this case 
however, instead of showing that the problem is in $\P/\poly$ we show that it satisfies
``$k$-membership comparability,"
a notion known to be related to solving $k$ instances using less than $k$ bits of help.

Next we consider the setting of average-case complexity:
Assume that we can solve $k$ instances of a decision problem 
using some help bits whose entropy is less than $k$ 
when the $k$ instances are drawn independently from a particular distribution.
Then we can efficiently solve an instance drawn from that distribution with probability better than $1/2$.

Finally, we show that in the case where $k$ is
super-logarithmic, assuming $k$-membership comparability of a decision problem, one cannot prove that the problem is in $\P/\poly$ by a
``black-box proof."
\end{abstract}


\section{Introduction}

In computational complexity theory, ``advice'' can be thought of as an extra input (in addition to the input instance) to algorithms that try to solve a computational problem. Advice depends only on the size of the input instance and captures the non-uniform model of computation. In this paper, we are interested in an extra input  that unlike advice depends on the particular instance of the problem to be solved. Such extra inputs (given as a string of bits) are called ``help bits" following \cite{Cai}.
Thus help bits are inputs of an algorithm
that similar to advice are generated by a \emph{trusted} computationally-unbounded party, but unlike advice depend on the instance of the problem and not only on its size.
We call the trusted party who provides us with the help bits the ``helper."

Help bits can be understood from the point of view of the oracle model as well.
Given an instance of the problem, we may ask an oracle about the help bits corresponding to that instance and then try to solve it.
Indeed, for $k$ bits of help, $k$ yes/no queries from an oracle are enough to find the help bits.
Here the oracle plays the role of the helper. Nevertheless, oracles can be used adaptively: A query to an oracle may depend on the answers to previous queries, but help bits cannot be requested adaptively and are given only at the beginning of the algorithm. This difference should be more carefully considered in the probabilistic as well as non-deterministic models of computation. The other issue here is the length of the help bits. For $k$ help bits, we should limit the use of the oracle to $k$ yes/no queries and not (say) polynomially many queries.

Any decision problem can be efficiently solved with one help bit
since we may ask the helper to give us the yes/no answer to the decision problem through that one bit of help.
Therefore, to obtain meaningful questions we consider the power of \emph{less than} one help bit per instance.
For example one may ask: Can we simultaneously solve $k$ instances of an NP-hard decision problem in polynomial time with $k-1$ help bits? That is, can we design this $k-1$ help bits cleverly in such a way that all the $k$ instances can be solved efficiently?

Suppose that there is an efficient algorithm that given $k$ instances $x_1,\dots, x_k$ of a decision problem $L$ and $\ell<k$ help bits, can output the bit-string $(L(x_1),\dots, L(x_k))$ where $L(x_i)=1$ if $x_i$ is in the language and $L(x_i)=0$ otherwise. The question is what can we do without having access to the true help bits? To this end,
assuming that $\ell$ is constant (or at most logarithmic), we may enumerate over all  $2^{\ell}$ possible help bits, and feed them to the algorithm. For each of these $2^{\ell}$ possible help bits, the output of the algorithm would be a bit-string of size $k$. Then we obtain a set of size at most $2^{\ell}$ which contains $(L(x_1),\dots, L(x_k))$. Equivalently, without having access to the true help bits, we have an efficient algorithm to reject $2^{k}-2^\ell$ possibilities for $(L(x_1),\dots, L(x_k))$.

The notion of \emph{$k$-membership comparability} defined in~\cite{Ogihara} is related to the above observation.
A decision problem is called $k$-membership comparable if there exists an efficient algorithm that for every $k$ instances $x_1,\dots, x_k$ computes $(b_1,\dots, b_k)$ such that $(L(x_1),\dots, L(x_k))\neq (b_1,\dots, b_k)$. Now we may ask the same question about $k$-member comparable sets: How hard can a $k$-membership comparable language be?

\vspace{.17in}
\noindent
\textbf{Previous works:}
Amir, Beigel and Gasarch~\cite{ABGconference, ABG}, though using the terminology of the bounded query model discussed above, have studied the above questions.\footnote{Here we should mention that before the work of Amir, Beigel and Gasarch, bounded query model had been used to classify functions (see e.g., \cite{Krentel}, \cite{Gasarch}, and \cite{GKR} for queries to NP oracles) and also had been studied in the theory of computability (see e.g., Beigel's thesis \cite{BeigelThesis}).}
They show that for constant $k$,
if $k$ instances of a decision problem can be solved efficiently using only $k-1$ help bits,
then the decision problem is in the class $\P/\poly$. Their result holds even with the weaker assumption of $k$-membership comparability, i.e., for constant $k$, a $k$-membership comparable set is in $\P/\poly$.

The authors of \cite{ABG} also show that for $k$ polynomial in the size of input instances,
$k$-membership comparable languages are in $\NP/\poly \cap \coNP/\poly$.\footnote{
We note that
in the presentation of \cite{BFP} of the results of \cite{ABG},
it is said that $k$-membership comparable languages for polynomial $k$ are in $\Sigma_2/\poly$.
However, Theorem 4.4 of \cite{ABG} implies the stronger result that the problem is in $\NP/\poly \cap \coNP/\poly$.
This stronger result in particular shows that NP-hard decision problems cannot be $k$-membership comparable for polynomial $k$ unless the polynomial hierarchy collapses.}

It is also shown in \cite{ABG} that for constant $k$,
if $2^k$ instances of a self-reducible decision problem can be efficiently solved using $k$ help bits, then the original decision problem is in the class $\P$.\footnote{This latter result in particular applies to NP-complete problems since SAT is self-reducible.
It also applies to some problems that are probably not NP-complete such as the decision version of integer factoring which is self-reducible,
and also to graph isomorphism which is equivalent to a self-reducible problem.}

Ogihara \cite{Ogihara}, Beigel, Kummer, and Stephan \cite{BKS}, and Agrawal and Arvind \cite{AA} show that if SAT is $k$-membership comparable for $k = c \lg_2 n$ where $c < 1$, then NP = P. Moreover, Sivakumar~\cite{Sivakumar} shows that if SAT is $O(\lg n)$-membership comparable, then NP = RP by solving Unique-SAT in deterministic polynomial time.

Classes of decision problems related to $k$-membership comparable problems include P-selective sets and verbose sets.
The class of P-selective sets is a sub-class of $2$-membership comparable languages, and
the notion of a verbose set is a generalization of $k$-membership comparability. See, for example,
\cite{OT} and \cite{Kobler} for some results on the power of oracles for such sets.

The computational value of help bits have  been considered in other models of computation as well, e.g., see \cite{NRS} and \cite{BH} for the decision tree model and \cite{Cai} for bounded depth circuits.

\vspace{.17in}
\noindent
\textbf{Our results:}
In the following we briefly discuss our contributions on the computational value of help bits in the two models of randomized and average-case computation.
We also discuss $k$-membership comparable problems for $k$ that is super-logarithmic.

\vspace{.17in}
\noindent
\textbf{Randomized complexity:}
Having $\ell$ help bits, we can always solve $k$ instances of a decision problem with probability at least $2^{\ell-k}$. This is simply because the $\ell$ help bits could be chosen to be the
answers to the first $\ell$ instances and we guess the answer of the rest of the instances randomly.
Our first result (Theorem~\ref{Thm:Random}) says that if for $k = O(\lg n)$, given $\ell$ help bits we can solve $k$ instances of size $n$ of a decision problem with probability non-negligibly better than $2^{\ell-k}$,
then the language is $k$-membership comparable.

Our next result (Theorem~\ref{Thm:RateDistortion}) is related to the classical problem of \emph{rate distortion} in information theory.
Suppose that instead of demanding correct answers for all of the $k$ instances, we only require that $(1-\alpha)$ fraction of our answers to the $k$ instances be correct. However the solver need not know which of the answers are correct. For instance, if the solver simply randomly guesses the answer to the $k$ instances, he can solve about $1/2$ of the instances with high probability for large enough $k$, without knowing exactly which of his guesses are correct. The question is how much better we can do with limited help bits.
	 By employing the rate distortion theory, it is possible to get the correct answers for at least $(1-\alpha)$ fraction of the $k$ instances with about $k(1 - h(\alpha))$ help bits, where $h(\cdot)$ is the binary entropy function.\footnote{The rate-distortion results   tells us that \emph{with high probability} $(1-\alpha)$ fraction of correct answers can be found with about $k(1 - h(\alpha))$ help bits.
Nonetheless, via a more refined argument, we show that for our particular case of interest, the same result can be obtained with probability one. Further this can be done to some degree efficiently.}
Our main result here is the converse, showing that for logarithmic $k$, unless the decision problem is $k$-membership comparable, we cannot solve $(1-\alpha)$ fraction of $k$ instances with fewer help bits, with any constant (independent of $k$) positive probability.

\vspace{.17in}
\noindent
\textbf{Average-case complexity:}
A computational problem may be easier on some instances compared to other instances.
Therefore, we may allow that harder instances receive more help,
e.g., harder instances may receive longer help bits.
But we still want to limit the average amount of information that the help bits may contain.
To this end we require that the help bits have limited information-theoretic \emph{entropy}.

For the definition of entropy to make sense, we need to have a distribution on the instances.
This way, we are in the realm of average-case complexity of the decision problem.
Our result here (Theorem~\ref{thm:mutualinformation}) says that if we start with a distribution on instances
such that one cannot guess the answer to an instance efficiently with probability quite more than $1/2$,
then we cannot solve $k$ instances efficiently with help bits whose entropy is less than $k$.
More precisely, 
we show that if $1/2 + \delta$ is the success probability of the best efficient algorithm on an average instance,
then minimum-entropy help bits for efficiently solving $k$ instances of the problem have between $k (1 - \Theta(\delta)) +  O(\lg k)$ and $k(1 - \Theta(\delta^2))$ bits of entropy.

\vspace{.17in}
\noindent
\textbf{$k$-Membership comparability for large $k$:}
Consider the predicate
\begin{center}
$\mathcal P(k)=$``If a decision problem is $k$-membership comparable, then it is in $\P/\poly$."
\end{center}
As discussed above $\mathcal P(k)$ it true for constant $k$ by \cite{ABG}. Here we first note (in Example~\ref{ex:superpolynomialk}) that $\mathcal P(k)$ is not true for super-polynomial $k$.
Second, we show (in Theorem~\ref{th:blackbox}) that for super-logarithmic $k$, we cannot prove $\mathcal P(k)$ in a ``black-box" manner, i.e., there is no proof of $\mathcal P(k)$ for super-logarithmic $k$ in which the algorithm that excludes a $k$-tuple for answers is used as a black-box. We emphasize that the proofs of all the previous results mentioned above have this black-box form.

This latter result is interesting since
\cite{Ogihara} first claimed that the proof of $\mathcal P(k)$ for constant $k$ by \cite{ABG} can be generalized to polynomial $k$,
but later (according to \cite{BFP}) retracted the claim.
Now, using our result, we know that the proof of \cite{ABG} did not generalize to polynomial $k$ for the stronger reason that it was a black-box proof.

\section{Preliminaries}
We begin by a formal definition of $k$-membership comparability.
This notion is the main topic of discussion in Section~\ref{Sec:membership-comparability}.
It is also used in the results of
Sections~\ref{Sec:randomized}.

For a decision problem $L$ and instance $x$ we let $L(x)=1$ if $x$ is in $L$, and $L(x)=0$ otherwise.

\begin{definition}[$k$-membership comparable sets]
A language $L$ is $k$-membership comparable if
given $k$ instances $(x_1, \ldots, x_k)$ each of size $n$,
we can find in time polynomial in $n$ and $k$, a bit-string $(b_1, \ldots, b_k) \in \{0,1\}^k$ such that
$(b_1, \ldots, b_k) \ne (L(x_1), \ldots, L(x_k)).$

The language $L$ is called non-uniform $k$-membership comparable
if the algorithm that finds $(b_1, \ldots, b_k)$ is allowed to be non-uniform,
i.e., it can receive polynomial advice.

The language $L$ is called randomized $k$-membership comparable
if given any $k$ instances $(x_1, \ldots, x_k)$ each of size $n$,
we can find in time polynomial in $n$ and $k$, a bit string $(b_1, \ldots, b_k) \in \{0,1\}^k$ such that
$$\Pr\big[(b_1, \ldots, b_k) \ne (L(x_1), \ldots, L(x_k))\big]  \ge (1 - 2^{-k})+1/{n^{O(1)}},$$
where the probability is only over the randomness of the algorithm.
\end{definition}
The notion of $k$-membership comparability is not new to this paper,
but the non-uniform and randomized versions of $k$-membership comparability had not appeared before.

Note that choosing a random vector $(b_1, \ldots, b_k)$, it is unequal to $(L(x_1), \ldots, L(x_k))$ with probability $1-2^{-k}$.
The definition of randomized $k$-membership comparability asks that we beat this probability with a non-negligible advantage.
Note also that according to our definition, randomized $k$-membership comparability is possible only when $k = O(\lg n)$.

In \cite{ABG}, it is proved that $k$-membership comparable problems, for constant $k$, are in $\P/\poly$.
Looking at their proof, it is clear that their result is true for non-uniform $k$-membership comparable problems as well:

\begin{theorem}[\cite{ABG}]
\label{ABGtheorem}
If a decision problem is non-uniform $k$-membership comparable for constant $k$,
then it is in $\P/\poly$.
\end{theorem}

In the following lemma,
we show that randomized $k$-membership comparable problems are non-uniform $k$-membership comparable.

\begin{lemma}\label{ABGprob}
Let $L$ be a decision problem.
If $L$ is randomized $k$-membership comparable,
then it is non-uniform $k$-membership comparable.
\end{lemma}
\begin{proof}
By assumption there is an algorithm that finds a vector $(b_1, \ldots, b_k)$ that is equal to $(L(x_1), \ldots, L(x_k))$ with probability $\epsilon \le 2^{-k} - 1/n^{O(1)}$.
If we iterate this algorithm for suitably $\poly(n)$ many times, then by Chernoff bound the vector $(b_1, \ldots, b_k)$ that has appeared with the most frequency
is equal to $(L(x_1), \ldots, L(x_k))$ with probability $< 2^{-nk}$.
Let $r$ be the randomness used in these many iterations.
Since there are only $2^{nk}$ vectors $(x_1, \ldots, x_k)$,
by the union bound
we can hardwire the randomness $r$ into the iterated algorithm
to get a polynomial-size circuit that outputs a vector $(b_1, \ldots, b_k) \ne (L(x_1), \ldots, L(x_k))$ given $(x_1, \ldots, x_k)$.
\end{proof}

In Section~\ref{Sec:randomized} we show that certain problems are randomized $k$-membership comparable.
By the above, they are also in $\P/\poly$ for constant $k$.

We will use the following direct product theorem in our proof of Theorem~\ref{thm:mutualinformation}.

\begin{lemma}
[\cite{GNW}]
\label{lemma:GNW}
Fix a language $L$ and a probability distribution $\cal D$ on $n$-bit strings.
Assume that circuits of size $s(n)$ cannot compute $L(x)$ for $x$ chosen according to $\cal D$ with probability better than $p(n) \in [1/2,1]$.
Then for any $\epsilon(n) > 0$,
circuits of size poly$(\epsilon(n)/n) s(n)$ cannot compute $(L(x_1), \ldots, L(x_k))$ for i.i.d.\ $x_1, \ldots, x_k \in \{0,1\}^n$,
where each $x_i$ is chosen according to $\cal D$,
with probability better than $p(n)^k + \epsilon$.
\end{lemma}

\vspace{.16in}
\noindent
\textbf{Entropy, mutual information:}
In the following we review the definition and properties of some information theoretic quantities that will be used in Section~\ref{Sec:information}.

For a discrete random variable $X$, the entropy of $X$ denoted by $H(X)$ is defined by
$$H(X)=\sum_x p(x)\lg_2\frac{1}{p(x)}.$$
Clearly,
\begin{align}
H(X) \ge \lg_2\frac{1}{\max_x p(x)}. \label{eq:minentropy}
\end{align}
When $X$ is a Bernoulli random variable with parameter $p$, the entropy function is written as
$$h(p)=p\lg\frac{1}{p}+(1-p)\lg\frac{1}{1-p}.$$
Fixing the size of the alphabet set of $X$ to be $m$, the maximum of $H(X)$ is attained at uniform distribution and is equal to $\lg(m)$.

For two discrete random variables $X, Y$, the conditional entropy $H(X|Y)$ is defined by
$$H(X|Y)=\sum_{y} p(y) H(X|Y=y),$$
where $H(X|Y=y)$ is the entropy of the random variable having distribution $p(x|y)$.
The conditional entropy $H(X|Y)$ is zero if $X$ is a \emph{function} of $Y$.
It is also easy to see that $H(X|Y) = H(X,Y) - H(Y)$.

Mutual information between two random variables $X$ and $Y$ is defined by
$$I(X;Y)=H(X)-H(Y|X)=H(Y)-H(Y|X)=H(X)+H(Y)-H(X,Y).$$
We always have $I(X;Y)\geq 0$.

Given three random variables, the conditional mutual information is defined by
$$I(X;Y|Z)=H(X|Z)-H(X|Y,Z)=H(Y|Z)-H(Y|X,Z).$$
We again have $I(X;Y|Z)\geq 0$.
Observe that $I(X;Y|Z)=0$ if for instance $X$ is a function of $Z$. The following equation which can easily be verified is known as the \emph{chain rule}:
$$I(X;YZ)=I(X;Y)+I(X;Z|Y).$$

\section{Randomized Complexity}
\label{Sec:randomized}
In the following theorem, we show
that
$\ell$ help bits cannot increase the probability
of correctly finding $(L(x_1), \ldots, L(x_k))$ by a factor of more than $2^l$
unless $L$ is non-uniform $k$-membership comparable.

\begin{theorem}\label{Thm:Random}
Let $\ell < k = O(\lg n)$ and $p \ge 2^{\ell-k} + 1/n^{O(1)}$.
If for a language $L$, we can correctly solve $k$ instances of size $n$ using $\ell$ help bits with probability $\ge p$,
then $L$ is randomized $k$-membership comparable.
\end{theorem}
\begin{proof}
Suppose that we guess the value of the $\ell$ help bits uniformly at random.
Then with probability $2^{-\ell}$ our guess of the help bits will be correct,
and hence with probability $q = 2^{-\ell}p\ge2^{-k}+1/n^{O(1)}$ we can correctly guess the answer to the $k$ instances without any help bits.

Assume that $(c_1, \ldots, c_k)$ is our guess for the $k$ instances.
If we choose $(b_1, \ldots, b_k)$ uniformly at random among all $2^k-1$ instances which are different from
$(c_1, \ldots, c_k)$, the probability that $(b_1, \ldots, b_k)$ is equal to the correct answer to the $k$ instances
is $(1 - q) /(2^k - 1) \le 2^{-k} - 1/n^{O(1)}$.
\end{proof}

As mentioned in the previous section using Theorem~\ref{ABGtheorem} and Lemma~\ref{ABGprob} we find that with the assumptions of the above theorem if $k$ is constant, then $L$ is in $\P/\poly$.

Rate distortion theory studies the minimum number of bits required for recovering, with some bounded distortion, a target bit-string. More precisely, to recover $(1-\alpha)$ fraction of a bit-string of length $k$ correctly we need about $k(1-h(\alpha))$ bits of information about that bit-string.
The following theorem shows that
we cannot decrease this minimum number of required bits
even if we know that the bits in the target bit-string
are the yes/no answers to $k$ instances of a hard decision problem.

\begin{theorem}\label{Thm:RateDistortion}
Let $p>0$ and $0<\alpha, \epsilon<1/2$ be some constants, and $L$ be some language. Also let $k = O(\lg n)$.
We are given $k$ arbitrary instances $x_1,\dots, x_k$ of $L$.
\begin{enumerate}
\item[\rm{(i)}] Assume that using $\ell< k(1-h(\alpha)-\epsilon)$ help bits, we can efficiently and with probability $p$ find a bit-string which matches $(L(x_1),\dots, L(x_k))$ in at least $(1-\alpha)$ fraction of positions. Then $L$ is randomized $k$-membership comparable if $k \ge \lg (3/(2p))/\epsilon$.

\item[\rm{(ii)}] Using $\ell\ge \lceil k(1-h(\alpha)+\epsilon) \rceil$ help bits, it is possible to efficiently find a bit-string which matches $(L(x_1),\dots, L(x_k))$ in at least $(1-\alpha)$ fraction of positions with probability one, if $k \ge 6(\epsilon^{-1} \lg(\epsilon^{-1}))$.
\end{enumerate}
\end{theorem}

\begin{proof} (i)  Without access to the help bits, we can try all the $2^{\ell}$ bit-strings of length $\ell$ for the help bits. For each such bit-string, we run the randomized algorithm many times (say $t$ times), forming a large table whose rows are $k$-bit strings. We denote these rows by $w_{ij}$ for $i\in\{1,2,\cdots, t\}, j\in\B^\ell$, i.e., each $w_{ij}$ is a row vector of $k$ bits. Therefore we have $2^{\ell}$ tables of size $t\times k$, with the $j$-th table having rows $w_{1j}, w_{2j}, \dots, w_{tj}$.

Let $j^*\in \B^\ell$ be the correct string of help bits.
Take $p' = 2p/3$, and $\delta = 2^{-(k+1)}$. By Chernoff bound, for sufficiently large (yet linear in $k$) $t$, with probability $1-\delta$, $p'$  fraction of rows $w_{1j^*}, w_{2j^*}, \dots, w_{tj^*}$ are within hamming distance $\alpha k$ from $(L(x_1),\dots, L(x_k))$. We will use this property to find some $(b_1,\dots, b_k)$ unequal to $(L(x_1),\dots, L(x_k))$.

For any $j\in \B^{\ell}$ define
$$M_j=\big\{v\in\{0,1\}^{k}: \exists I\subseteq \{1,2,\cdots, t\}: |I|\geq p't, d_H(v, w_{ij})\leq \alpha k,~\forall i\in I\big\},$$
where $d_H(\cdot, \cdot)$ denotes the hamming distance. By the above discussion with probability $\ge1-\delta$, the correct answer $(L(x_1),\dots, L(x_k))$ belongs to $M_{j^*}$.
Observe that $M_j$'s can be computed efficiently because $2^k$ and $t$ are of size polynomial in $n$.

We claim that for any $j$,
\begin{align}\label{eq:bound-mj}
|M_j|\leq  2^{k(h(\alpha)+\epsilon)}.
\end{align}
To this end consider a bipartite graph with vertex set $\{1,2, \dots, t\}$ for one part and vertex set $M_{j}$ for the other. We draw an edge $i\sim v$ for $i\in\{1,2, \dots, t\}$ and $v\in M_j$  if $d_H(w_{ij}, v)\leq \alpha k$. Then by the definition of $M_{j}$, for any $v\in M_{j}$
we have $\text{deg}(v)\geq p't$.
On the other hand, the number of sequences that are within hamming distance $\alpha k$ of any given bit-string of size $k$  is at most
$2^{kh(\alpha)}$ (see e.g., \cite[page 310, Lemma 8]{MW}). This means that for every $1\leq i\leq t$ we have $\text{deg}(i)\leq 2^{kh(\alpha)}$. Double counting the number of edges of the bipartite graph we obtain
$$p't|M_{j}| \leq\sum_{v\in M_{j}}  \text{deg}(v)=\sum_i \text{deg}(i)\leq t2^{kh(\alpha)}.$$
This gives~\eqref{eq:bound-mj} if $k \ge \lg (1/p')/\epsilon$.

Now  using equation \eqref{eq:bound-mj} and the union bound, we get that 
$$\bigg| \bigcup_j {M}_j \bigg|\leq 2^\ell2^{k(h(\alpha)+\epsilon)}<2^k.$$
Therefore we can find $(b_1,\dots, b_k)\in \B^k$ that does not belong to $\bigcup_j {M}_j$. Since with probability at least
$1- \delta = 1 - 2^{-k} + 2^{-k-1} = 1 - 2^{-k} + 1/n^{O(1)}$ we have $(L(x_1),\dots, L(x_k))\in M_{j^*}$, then $(b_1, \dots, b_k)\neq (L(x_1),\dots, L(x_k))$ with probability $1 - 2^{-k} + 1/n^{O(1)}$.
Thus, $L$ is randomized $k$-membership comparable.\\

(ii)
We claim that there is a subset $R\subseteq \B^k$ of size $|R|=2^{k(1 - h(\alpha)+\epsilon)}$ such that for any $v\in\{0,1\}^k$ there is $w\in R$ such that $d_H(v, w)\leq \alpha k$. This requires a more  refined argument than the one given by the standard rate distortion theory which is only telling us that for \emph{most} sequences $v\in\{0,1\}^k$ there is $w\in R$ such that $d_H(v, w)\leq \alpha k$. 
Furthermore,  we show that  we can construct $R$ in time polynomial in $2^k$. Assuming this, the helper can give the sequence $w\in R$ whose hamming distance from $(L(x_1),\dots, L(x_k))$ is at most $\alpha k$, and the algorithm may just output this sequence. Note that the helper would need $\lceil \lg |R| \rceil \leq \ell$ bits to send $w$.

For a point $v \in \{0,1\}^k$ and $r$,
let $B(v,r)$ denote the hamming ball of radius $r$ around $v$.  With this definition, we are looking for a set $R$ of points in $\{0,1\}^k$ whose union of hamming balls cover the entire set  $\{0,1\}^k$. 
To construct the set $R$,
we start with an empty set and  use a greedy algorithm to add points $v^1, \ldots, v^{|R|}$ in this order as follows:
Having added points $v^1, \ldots, v^j$ for
$0 \le j < |R|$,
we choose point $v^{j+1} \in \{0,1\}^k$ such that
$B(v^{j+1}, k\alpha)$ covers
the most number of points not covered by $B(v^1, k\alpha), \ldots, B(v^j, k\alpha)$. In case there are multiple choices of $v^{j+1}$ with this property,
choose $v^{j+1}$ to be the lexicographically smallest one.

Let $q=q(k, \alpha k)$ be the probability that a random point is in hamming distance $k\alpha$ of a given point.
Any point not in $\bigcup_{i=1}^j B(v^i, k\alpha)$ is covered by $q$ fraction of balls of radius $k\alpha$. Then there is $B(v, k\alpha)$ that covers at least $q$ fraction of points not in $\bigcup_{i=1}^j B(v^i, k\alpha)$. 
Based on this observation and by a simple induction we find that 
the number of points not covered by $B(v^1, k\alpha), \ldots, B(v^{|R|}, k\alpha)$ is
$\le 2^k (1 - q)^{|R|}$. 
By \cite[page 310, Lemma 8]{MW},
we have $$q \ge \frac{2^{k(-1+h(\alpha))}}{\sqrt{8k\alpha(1-\alpha)}} \ge \frac{2^{k(-1+h(\alpha))}}{\sqrt{2k}}.$$
When $k \geq 6\epsilon^{-1} \lg(\epsilon^{-1})$,
we have $2^k(1 - q)^{|R|} \le 2^k e^{-q |R|} < 1$. Hence no point in $\{0,1\}^k$ remains uncovered.
 We are done.
\end{proof}

Again with the assumption of part (i) of this theorem, if $k$ is constant then $L$ is in $\P/\poly$.

We note that given $k$ instances of a computational problem, the task of
solving at least a given fraction of the instances correctly had been studied before
(e.g., see \cite{HW}).
What is new to our work (Theorem~\ref{Thm:RateDistortion}) is allowing that the task be done
using some limited help bits and also with some bounded probability of error.

\section{Average-case Complexity}
\label{Sec:information}
As motivated earlier, 
the length of the help bits may depend on the instances of the problem since for easy instances we do not need any help. In this section we  measure the amount of help on average over all choices of the $k$ instances. For this to make sense we need a distribution on instances, so we fall in the setting of average-case complexity.

Let us fix some notation before stating our results. Let $\mathcal D_n$ be a distribution on instances of length $n$. We assume that the $k$ instances $x_1,\dots, x_k$ are drawn i.i.d.\ according to $\mathcal D_n$. We denote the random variables associated to these randomly chosen instances by $X_1,\dots, X_k$. Let $\ell$ be the maximum length of the help bits. 
This means that the alphabet set of the help bits is $\Sigma=\bigcup_{i=0}^{\ell}\{0,1\}^i$. Here we allow $\ell$ to be at most $\ell = O(\lg n)$ so that the size of its alphabet set be $|\Sigma|=n^{O(1)}$.
Since $x_1, \dots, x_k$ are chosen randomly, the string of help bits is a random variable too. We denote this random variable 
by $S=S(X_1,\dots, X_k)$ which takes its values in $\Sigma$.

We may measure the amount of the help bits in two different ways: entropy of the help bits random variable, and its average length. 
In the following we first consider entropy, and later comment on average length. 
 
\begin{theorem}
\label{thm:mutualinformation}
Let $L$ be a language.
For every integer $n$ fix a distribution
$\mathcal D_n$ on instances of size $n$, i.e., on $\B^n$.
Let $k=O(\lg n)$, and let $\Sigma=\bigcup_{i=0}^{\ell}\{0,1\}^i$ such that $|\Sigma| = n^{O(1)}$. Suppose that there is an efficient algorithm such that for every $x_1,\dots, x_k\in \B^n$ and using help bits $s=s(x_1,\dots, x_k)\in \Sigma$  outputs $(L(x_1), \dots, L(x_k))$.
Assume further that for $1/n^{O(1)} < \delta(n,k) < 1/2$,
$$H(S)\leq k - 3 k \delta,$$
where $S=S(X_1,\dots, X_k)$ is the random variable corresponding to the help bits of the randomly chosen $x_1,\dots, x_k\in \B^n$ according to $\mathcal D_n$. Then there exist polynomial size circuits that correctly solve a randomly chosen $x\in \B^n$
with probability
$\ge 1/2 + \delta$ without having access to the help bits.
\end{theorem}

Notice that by hardness-amplification techniques such as \cite{O'Donnell} and \cite{HVV},
it is plausible to find problems in NP with average-case hardness such that
it is not possible to solve them efficiently with probability better than $1/2 + o(1)$.
For such problems, the above theorem becomes relevant.\\

\begin{proof}
Let $\{C_n: n\geq 1\}$ be the polynomial size circuits that solve $k$ instances of $L$ with help bits, i.e., for all $x_1,\dots, x_k\in \B^n$ and the appropriate $s\in \Sigma$ we have $C_n(x_1,\dots, x_k, s)= (L(x_1),\dots, L(x_k))$. In the following we drop the subscript $n$ in $C_n$ and $\mathcal D_n$ as they are clear from the context. We prove the theorem in a few steps.\\

\noindent
(Step 1) Let us define
$$t=t(x_1,\dots, x_k)=\big(C(x_1,\dots, x_k, \sigma_1), \dots, C(x_1,\dots, x_k, \sigma_{|\Sigma|})\big),$$
where $\Sigma=\{\sigma_1,\dots, \sigma_{|\Sigma|}\}$, and let $T$ be its corresponding random variable (when $x_i$'s are chosen randomly).
Note that $t$ can be computed by a polynomial size circuit since $|\Sigma|=n^{O(1)}$.  We claim that
$$H\big(L(X_1), \ldots , L(X_k)| T\big) \le k - 3 \delta k.$$

For any $x_1,\dots, x_k$, if we are given the true help symbol $s=\sigma_i$, we can compute $(L(x_1), \dots , L(x_k))$ from $t$ by taking its $i$-th element.
Thus
\begin{align}
H\big(L(X_1), \dots, L(X_k)| S, T\big)=0.
\label{eqn:a1}
\end{align}

Therefore,
\begin{align}
&H(L(X_1), \ldots, L(X_k)) \nonumber
\\&=I(L(X_1), \ldots, L(X_k); S, T)\label{eqn:a2}
\\&=I(L(X_1), \ldots, L(X_k); T)+I(L(X_1), \ldots, L(X_k); S|T)\label{eqn:a4}
\\&\leq I(L(X_1), \ldots, L(X_k); T)+I(L(X_1), \ldots, L(X_k),T; S)\label{eqn:a5}
\\&\leq I(L(X_1), \ldots, L(X_k); T)+I(X_1, \ldots, X_k, L(X_1), \ldots, L(X_k),T; S)\nonumber
\\&= I(L(X_1), \ldots, L(X_k); T)+I(X_1, \ldots, X_k; S)+I(L(X_1), \ldots, L(X_k),T; S|X_1, \ldots, X_k)\nonumber
\\&= I(L(X_1), \ldots, L(X_k); T)+I(X_1, \ldots, X_k; S)\label{eqn:a3}\\
&\le \big(H(L(X_1), \ldots, L(X_k)) - H(L(X_1), \ldots, L(X_k)| T)\big) +H(S)\nonumber\\
&\le \big(H(L(X_1), \ldots, L(X_k)) - H(L(X_1), \ldots, L(X_k)| T)\big) +k- 3 \delta k. \nonumber
\end{align}
Here
\eqref{eqn:a2} follows from \eqref{eqn:a1} and the definition of mutual information.
\eqref{eqn:a4} and~\eqref{eqn:a5} follow from the chain rule.
Moreover, \eqref{eqn:a3} follows from the fact that $H(L(X_1),\ldots, L(X_k),T|X_1, \ldots, X_k)=0$.
We conclude that
\begin{align}\label{eq:ikx}
H(L(X_1), \dots, L(X_k) | T) \le k - 3 \delta k.
\end{align}\\

\noindent
{(Step 2)}
We have
\begin{align}
&\sum_{t}\Pr(T=t)\bigg[\max_{c_1, \ldots, c_k}\Pr(L(X_1)=c_1, \ldots, L(X_k)=c_k|T=t)\bigg] \nonumber \\
&\ge\sum_{t}\Pr(T=t)2^{-H(L(X_1), \dots, L(X_k) | T = t)} \label{eqn:useminentropy} \\
&\ge2^{-\sum_{t}\Pr(T=t)H(L(X_1), \dots, L(X_k) | T = t)} \nonumber\\
&=2^{-H(L(X_1), \dots, L(X_k) | T)} \nonumber\\
&\ge2^{-k + 3 \delta k},\label{eq:step2}
\end{align}
where \eqref{eqn:useminentropy} follows from equation \eqref{eq:minentropy}.\\

\noindent
{(Step 3)}
In this step, we show how we can construct binary random variables $B_1, \ldots, B_k$ using polynomial-size circuits such that $\Pr[L(X_1) = B_1, \ldots, L(X_k) = B_k] \ge 2^{-k + 3 \delta k}.$

For any $t$, let $v_t$ denote the vector $(c_1, \ldots, c_k) \in \B^k$
that maximizes $\Pr[L(X_1)=c_1,\dots, L(X_k)=c_k| T=t]$. Observe that there are at most $\poly(n)$ possible $t$'s for any $n$  since $k=O(\lg n)$.
Therefore, we can give $v_t$ for all $t$ as advice to our algorithm.

Now given $x_1, \ldots, x_k$, we can compute $t=t(x_1,\dots, x_k)$ as in Step 1 using a polynomial size circuit (because $|\Sigma| = n^{O(1)}$).
Then we give $v_t=(b_1,\dots, b_k)$ as the solution for $(L(x_1), \ldots, L(x_k))$. Inequality~\eqref{eq:step2} basically says that
\begin{align*}\label{eq:bk}
\Pr(L(x_1) = b_1, \ldots, L(x_k) = b_k)
\ge 2^{-k + 3 \delta k}.
\end{align*}\\

\noindent
{(Step 4)}
Assume that we cannot efficiently solve an instance chosen according to distribution $\mathcal D$
with probability better than $1/2 + \delta$.
By the direct product lemma of Goldreich, Nisan, and Wigderson (Lemma \ref{lemma:GNW}),
we cannot compute $(L(X_1), \ldots, L(X_k))$ with probability better than $(1/2 + \delta)^k + 1/p(n)$ for any polynomial $p$ in time polynomial in $n$.
But
\begin{align*}
2^{(-1 + 3 \delta)k} - (1/2 + \delta)^k & \ge 2^{-k} ((2^{3\delta})^k - (1 + 2\delta)^k) \\
&=  2^{-k} ((e^{3 \ln(2) \delta})^k - (1 + 2\delta)^k) \\
& \ge  2^{-k} ((1 + 3 \ln(2) \delta)^k - (1 + 2\delta)^k) \\
& \ge  2^{-k} (3 \ln(2) - 2) \delta (1 + 2\delta)^{k-1} \\
& \ge  1/n^{O(1)},
\end{align*}
because both $2^{-k}$ and $\delta$ are $1/n^{O(1)}$.
Therefore, we cannot compute $(L(X_1), \ldots, L(X_k))$ with probability as good as $2^{-k + 3 \delta k}$.
This is in contradiction with Step 3.
We are done.
\end{proof}

We may measure the amount of help by its mutual information (as a measure of correlation) with the input instances. In this case the same theorem as above holds when we replace the assumption $H(S)\leq k(1 - 3 \delta)$ with $I(X_1, \dots, X_k; S)\leq k(1 - 3 \delta)$. 
Indeed in the above proof we use $ I(X_1, \dots, X_k; S)\leq H(S)\leq k(1 - 3 \delta)$ in Step 1, which holds if we assume $I(X_1, \dots, X_k; S)\leq k(1 - 3 \delta)$ in the first place.

The following proposition is in the converse direction of Theorem \ref{thm:mutualinformation}.

\begin{proposition}
\label{thm:proposition}
Let $L$ be a language. Suppose that there exist polynomial size circuits $C_n$ that correctly solve a randomly chosen $x\in \B^n$
with probability
$\ge 1/2 + \delta$ without having access to any help bits. Then for any $k\in\mathbb{N}$, there are help bits $s=s(x_1,\dots, x_k)\in \{0,1\}^k$ and a polynomial-size circuit that outputs $(L(x_1), \dots, L(x_k))$ using these help bits, such that
$$H(S)\leq kh(1/2 + \delta)=k(1-O(\delta^2)),$$
where $h(\cdot)$ is the binary entropy function.
\end{proposition}

\begin{proof} Given an instance $x$, let $B(x)=L(x)\oplus C_n(x)$ where $\oplus$ denotes summation modulo $2$. Then by the assumption $\Pr[B(X)=0]\geq 1/2 + \delta$ where $X$ is the random variable corresponding to the randomly chosen $x$. By the monotonicity of $h(\cdot)$ on $[1/2, 1]$ we conclude that $H(B(X))\leq h(1/2 + \delta)$.

Now given $k$ instances $x_1, \cdots, x_k$, consider the help bits 
$$s=\big(B(x_1), \cdots, B(x_k)\big).$$
Observe that using $s$ we can efficiently compute $L(x_1), \cdots, L(x_k)$ as follows: We first use $C_n$ to compute $(C_n(x_1), \cdots, C_n(x_k))$ and then XOR it with $s$ to recover $L(x_1), \cdots, L(x_k)$.
Moreover  the entropy of the random variable associated with $s$ is
\begin{align*}
H(S)&=\sum_{i=1}^k H( B(X_i))\leq k\cdot h(1/2 + \delta).
\end{align*}
We are done.
\end{proof}

Let $L$ be an arbitrary language and fix a distribution on its input instances as above. Let $\delta\geq 0$ be the largest number such that a randomly chosen $x$ can be correctly solved with polynomial size circuits with probability $1/2+\delta$. Now suppose that we want to solve $k$ randomly chosen instances. 
Theorem \ref{thm:mutualinformation} says that to solve these $k$ instance correctly we need help bits of entropy at lease 
$k-3 \delta k=k(1-\Theta(\delta))$. On the other hand, Proposition~\ref{thm:proposition} says that there are help bits with entropy $k h(1/2 + \delta)=k(1-\Theta(\delta^2))$ which given them we can solve the $k$ instances efficiently. 
Here the question is which of the two bounds are tight.
The following example shows that in some cases the lower bound is closer to the actual required entropy.

\begin{example}
\label{example:O(delta)}
Let $L$ be a random decision problem defined on instances of size $n$ as follows. Let $F$ be the lexicographically first $2\delta$ fraction of instances. 
 We let $L(x)=0$ if $x\in F$, and otherwise we choose $L(x)\in \{0,1\}$ 
uniformly at random.
Then with high probability over the random choice of $L$, no polynomial size circuit can solve the decision problem on more than $1/2 + \delta + 2^{-\Theta(n)}$ fraction of the instances.
On the other hand, we can solve $k$ uniformly random independent instances of the problem using help bits with entropy $\le k (1 - 2 \delta) + \lg_2 (k+1)$.
\end{example}

\begin{proof}
The first part, that the average-case complexity of the problem is high, can be shown by a counting argument similar to \cite{Shannon}.

For the second part we define the help bits as follows: 
For an instance $x$,
let $B(x)$ be a string of length at most 1, 
where $B(x)$ is the empty string when $x\in F$,
and $B(x)=L(x)$ otherwise.
Now let the help bits for $k$ instances $x_1, \ldots, x_k$ be $s = (B(x_1), \ldots, B(x_k))$.

Given $x_1, \ldots, x_k$ we can check which of $x_1, \ldots, x_k$ are among the $2 \delta$ lexicographically first fraction of instances.
We trivially know the answer to these instances.
For the rest of the instances, we can read off their answers from the help bits
(since we know for which values of $i$ the string $B(x_i)$ is empty.)

Let $S$ denote the random variable associated to $s$, and let $M$ be its length. Since $0\leq M\leq k$, and $M$ is a function of $S$ we have
\begin{align*}
H(S)&=H(S,M)\\
&= H(M) + H(S| M) \\
&\leq \lg_2(k+1) + \sum_{j=0}^k \Pr[M=j] H(S | M=j) \\
&\leq \lg_2(k+1) + \sum_{j=0}^k \Pr[M=j] j\\
&= \lg_2(k+1) + k(1 - 2\delta).
\end{align*}

\end{proof}

In the above example, we used the fact that there is a $1 - 2 \delta$ fraction hard-core set of instances 
and the rest of the instances are easy.
In general, by Impagliazzo's hard-core set lemma \cite{Impagliazzo},
such a hard-core set always exists.
However, we cannot utilize this hard-core set, 
since unlike the example above
there does not exist a general procedure to test which instances are in the hard-core set.

We now turn our attention to measuring the amount of help by its average length. As above, let $S$ be the random variable corresponding to help bits, and let $M$ be its length. Here the crucial observation is that with no loss of generality we may assume that $M\leq k$. Indeed if for some $x_1, \dots, x_k$ the length of $s=s(x_1, \dots, x_k)$ is larger than $k$, then we may replace $s$ with $(L(x_1), \dots, L(x_k))$.
There is a problem here if we do this conversion for one particular sequence: the solver should know that the $k$-bit help is $(L(x_1), \dots, L(x_k))$. For this reason, we should do the conversion for \emph{all} sequences whose help bit string has length at least $k$, so that the solver knows what to do whenever he sees such a sequence. Clearly these new help bits are as useful as $s$, and have a shorter length. Note that this procedure may increase the entropy of $S$ because in computing the entropy, it is only the probabilities assigned to sequences that matter, not their length. 

\begin{corollary} 
Use the notation of Theorem~\ref{thm:mutualinformation}. Assume that 
$$\mathbb E[M]\leq k - 3 k \delta-\lg(k+1),$$
We further assume that $M\leq k$. Then there exist polynomial size circuits that correctly solve a randomly chosen $x\in \B^n$
with probability
$\ge 1/2 + \delta$ without having access to the help bits.
\end{corollary}

\begin{proof} 
We compute
\begin{align*}
H(S)&= H(S, M)\\
&\leq H(M)+ H(S|M)
\\&\leq \lg(k+1)+ \sum_{j=0}^{k}\Pr[M=j]H(S|M=j)
\\&\leq \lg(k+1)+ \sum_{j=0}^k \Pr[M=j]j
\\&= \lg(k+1)+ \mathbb{E}[M]\\
&\leq k-3k\delta.
\end{align*}
The result then follows by applying Theorem~\ref{thm:mutualinformation}.
\end{proof}

\section{$k$-Membership Comparability for Large $k$}
\label{Sec:membership-comparability}
In this section we study $k$-membership comparability for non-constant values of $k$.
More precisely, we look at extensions of Theorem~\ref{ABGtheorem} for larger $k$.
In the following we first show that $k$-membership comparable languages do not necessarily have polynomial-size circuits.

\begin{example}
\label{ex:superpolynomialk}
For every function $k=k(n) = n^{\omega(1)}$, there exists a decision problem which is $k$-membership comparable but is not in $\P/\poly$.
\end{example}

\begin{proof}
Let $k'(n)=2^{\ell(n)}$ be the largest power of $2$ such that $k'\leq \min\{k-1, 2^n\}$.
Then $k' = n^{\omega(1)}$.
Let $f_n:\{0,1\}^{\ell} \rightarrow \{0,1\}$ be a sequence of functions that cannot be computed by circuits of size $O(k'/\ell)$, (see \cite{Shannon}). Now let $L$ be the language that for $x=(b_1,\dots, b_n)\in \B^n$ is given by $L(x) = f(b_1, \ldots, b_{\ell})$.
Then clearly $L$ is not in $\P/\poly$.

We now show that $L$ is $k$-membership comparable.
Given $x_1, \ldots, x_k \in \{0,1\}^n$,
there are $i \ne j$ such that $x_i$ and $x_j$ coincide on their first $\ell$ bits. Such indices $i,j$ can be found efficiently.
Then, we know that $(L(x_1), \ldots,L(x_k))$ is not equal to the bit-string
that has $0$ in position $i$, and $1$ in position $j$ (and arbitrary values in other positions).
\end{proof}

Whether Theorem~\ref{ABGtheorem} can be extended all the way to $k=\poly(n)$ is an open question.
In the following we argue that for super-logarithmic $k$,
no ``black-box proof" shows $k$-membership comparable problems are in $\P/\poly$.
Note that the proof of Theorem~\ref{ABGtheorem} for constant $k$ is a black-box proof.

\begin{theorem}
\label{th:blackbox}
For  $k(n) = \omega(\lg n)$
we cannot show in a black-box way that $k$-membership comparability of a decision problem $L$ implies $L\in\P/\poly$.
In other words, there exist a language $L$ and an oracle $O$ that given $x_1, \ldots, x_k \in\{0,1\}^n$ outputs a bit-string $(b_1, \ldots, b_k) \ne (L(x_1), \ldots, L(x_k))$, but no algorithm can compute $L$ using oracle $O$ in $\poly(n)$ time even given $\poly(n)$ bits of advice (that can depend on both $L$ and $O$).
\end{theorem}

\begin{proof}
Consider the restrictions of $L$ and $O$ on inputs of size $n$, i.e., $L_n=L\cap\{0,1\}^n$ and $O_n=O\cap(\B^n)^k$.
There are $2^{2^n}$ possibilities for $L_n$, and once $L_n$ is fixed there are $(2^k-1)^{2^nk}$ for the value of $O$ on $k$-tuples of $n$-bit strings.
In total, the number of acceptable  pairs $(L_n, O_n)$ is $2^{2^n} (2^k-1)^{2^nk}$.

Assume that for each pair $(L, O)$, there exists an algorithm that computes $L$ using $s(n) = n^{O(1)}$ number of queries to $O$ and $a(n) = n^{O(1)}$ bits of advice.
We will use this assumption to upper-bound the number of acceptable pairs $(L_n,O_n)$.

Fix a particular advice. Suppose that we want to compute $L(x)$ for all $x\in \B^n$. For each such instance we make $s(n)$ queries from $O$, so in total the number of all queries is $q(n) \le 2^n s(n)$.
Each of these queries (once $L$ is not fixed) has $2^k$ possible answers.
Therefore, for a fixed value of advice, there are $(2^k)^{q(n)}$ possible query answers for these $q(n)$ queries.
(Note that this is true even if the queries may be adaptive.) Now fixing these query answers, $L_n$ becomes fixed too. Then each of the remaining $2^{nk} - q(n)$ possible queries from the oracle have only $2^k-1$ valid answers.
Since there are $2^{a(n)}$ possible advices,
the number of pairs $(L_n, O_n)$ is at most
$$2^{a(n)} (2^k)^{q(n)} (2^k-1)^{2^{nk}-q(n)} = 2^{a(n)} 2^{\Theta(2^{-k}s(n)2^n)} (2^k-1)^{2^{nk}} < 2^{2^n} (2^k-1)^{2^{nk}},$$
where we used $2^{-k} s(n) = o(1)$.
This last inequality is in contradiction with the number of pairs $(L_n, O_n)$ computed above.
\end{proof}

\section{Open Problems}
We studied the computational benefit of help bits in some new settings, namely in randomized and average-case complexity frameworks. We also discussed $k$-membership comparability for superlogarithmic $k$.  

We mention a few open problems here:
In Theorem \ref{thm:mutualinformation}, 
it is open whether we can get non-trivial results when $\delta$ is close to $1/2$.
For that to work, we need to improve the assumption $H(S) \le k - 3k\delta$ to something like 
 $H(S) \le k - 2k\delta$.

Another open problem is classifying decision problems for which Theorem \ref{thm:mutualinformation} is tight vs.\ decision problems for which 
Proposition \ref{thm:proposition} is tight.
We already gave an example that shows Theorem \ref{thm:mutualinformation} is tight for some cases.

Finally, it is interesting to better understand the relationship between $k$-membership comparability and non-uniform complexity,
especially for logarithmic $k$.

\end{document}